\documentclass[11pt,a4paper]{article}
\usepackage{jcappub}
\def\gsim{\;\rlap{\lower 2.5pt
 \hbox{$\sim$}}\raise 1.5pt\hbox{$>$}\;}
\def\lsim{\;\rlap{\lower 2.5pt
   \hbox{$\sim$}}\raise 1.5pt\hbox{$<$}\;}
\def\ie{{\it i.e. }}
\def\eg{{\it e.g. }}

\def\vtl{\vec{x}^L}
\def\vts{\vec{x}^S}

\title{Measuring the Reduced Shear}

\author{Jun Zhang}
\affiliation{Texas Cosmology Center, the University of Texas at Austin, \\
Austin, TX 78712, USA} 
\affiliation{Department of Astronomy, University of California, \\
Berkeley, CA 94720, USA}
\emailAdd{jzhang@astro.as.utexas.edu}

\abstract{Neglecting the second order corrections in weak lensing measurements can lead to a few percent uncertainties on cosmic shears, and becomes more important for cluster lensing mass reconstructions. Existing methods which claim to measure the reduced shears are not necessarily accurate to the second order when a point spread function (PSF) is present. We show that the method of Zhang (2008) exactly measures the reduced shears at the second order level in the presence of PSF. A simple theorem is provided for further confirming our calculation, and for judging the accuracy of any shear measurement method at the second order based on its properties at the first order. The method of Zhang (2008) is well defined mathematically. It does not require assumptions on the morphologies of galaxies and the PSF. To reach a sub-percent level accuracy, the CCD pixel size is required to be not larger than $1/3$ of the Full Width at Half Maximum (FWHM) of the PSF, regardless of whether the PSF has a power-law or exponential profile at large distances. Using a large ensemble ($\gtrsim 10^7$) of mock galaxies of unrestricted morphologies, we study the shear recovery accuracy under different noise conditions. We find that contaminations to the shear signals from the noise of background photons can be removed in a well defined way because they are not correlated with the source shapes. The residual shear measurement errors due to background noise are consistent with zero at the sub-percent level even when the amplitude of such noise reaches about $1/10$ of the source flux within the half-light radius of the source. This limit can in principle be extended further with a larger galaxy ensemble in our simulations. On the other hand, the source Poisson noise remains to be a cause of systematic errors. For a sub-percent level accuracy, our method requires the amplitude of the source Poisson noise to be less than $1/80\sim 1/100$ of the source flux within the half-light radius of the source, corresponding to collecting roughly $10^4$ source photons.}

\keywords{cosmology, large scale structure, gravitational lensing - methods, data analysis - techniques, image processing}

\arxivnumber{1002.3614}
\newtheorem{theorem}{Theorem}[section]

\newenvironment{proof}[1][Proof]{\begin{trivlist}
\item[\hskip \labelsep {\bfseries #1}]}{\end{trivlist}}
\newcommand{\qed}{\nobreak \ifvmode \relax \else
      \ifdim\lastskip<1.5em \hskip-\lastskip
      \hskip1.5em plus0em minus0.5em \fi \nobreak
      \vrule height0.75em width0.5em depth0.25em\fi}

\begin{document}

\maketitle
\flushbottom

\section{Introduction}
\label{intro}

Weak gravitational lensing has been widely used as a direct probe of the mass distribution of our Universe on different scales, including large scale structure, clusters, galaxies, etc. \cite{hj08}. Not only is the physics of lensing well understood in the context of General Relativity, but the lensing effect can also be straightforwardly measured using the shapes of background galaxy images \cite{kwl00,vw00,wittman00}. Currently, one of the main challenges in this field is about how to accurately recover the cosmic shear field from galaxy shapes \cite{heymans06,massey07,bridle09a,bridle09b}. This is difficult due to the large galaxy shape noise, the involvement of the point spread function (PSF), the presence of the photon noise, the pixelation effect, etc.. There have been many literatures focusing on this particular topic \cite{tyson90,bonnet95,kaiser95,luppino97,hoekstra98,rhodes00,kaiser00,bridle01,bernstein02,refregierbacon03,massey05,kuijken06,miller07,nakajima07,kitching08,zhang08,zhang09}. 

In all of the practical shear measurement methods proposed so far, there is a common assumption: the cosmic shear is small, therefore the second or higher order terms in shear can be neglected. This is true for the shear field of our Universe on large scales, which is typically of order a few percent. However, future weak lensing survey may require shear measurement accuracy to be controlled below a $0.1\%$ level \cite{htbj06,ar08}. On arc minute angular scales, the second order terms can cause a systematic error of order $10\%$ to the cosmic shear power spectrum \cite{dsw06,shapiro09}. More importantly, the shear field by a foreground cluster can easily be of order ten percent or more. Neglecting second order terms in shear measurements can lead to significant errors on the implied cluster masses \cite{wittman01,hoekstra01,gray02,taylor04,broadhurst05,leonard07,heymans08,deb08}. 

The main purpose of this paper is to further develop the shear measurement method of \cite{zhang08}(Z08 hereafter) by including the second order terms in shear/convergence in the formalism. This is done straightforwardly in \S\ref{method}, in which we show that to the second order in accuracy, Z08 measures exactly the reduced shears. A popular misunderstanding in the weak lensing community is that all of the existing shear measurement methods are already accurate to the second order in shear/convergence, because they all claim to measure the reduced shears. We show why this is not generally true in the presence of PSF. In \S\ref{theorem}, we provide a simple and useful theorem for judging whether shear estimators are accurate to the second order based on their properties at the first order. The theorem provides an easy way to see why the method of Z08 yields exactly the reduced shear. \S\ref{test} demonstrates the accuracy of Z08 using a large number of computer-generated mock galaxies of unrestricted morphologies in the presence of PSF and the photon noise, including both the background noise and the source Poisson noise. Finally, we conclude in \S\ref{summary}.

%-------------------------------------------------------------------
\section{The Shear Measurement Method}
\label{method}

The basic idea of Z08 is to use the spatial derivatives of the galaxy surface brightness field to measure the cosmic shear. In a parallel paper \cite{zhang10}, we have shown that this method is equivalent to measuring shears using the galaxy quadrupole moments in Fourier space, with an additional term correcting for the PSF. Indeed, the measurement should be carried out in Fourier space, in which the moments can be easily evaluated, and the PSF can be transformed into the desired isotropic Gaussian form through multiplications. Further more, as shown in the parallel paper, not only the measurement, but also the whole analytic derivation of the relation between the galaxy surface brightness and the cosmic shear can be worked out in Fourier space in a much simpler way than in real space. As will be shown below, because of this convenience, the relation can be made accurate even to the second order in shear. 

Let us start the discussion in real space. We define the galaxy surface brightness distribution before lensing as $f_S(\vec{x}^S)$, the lensed galaxy (before being processed by the PSF) as $f_L(\vec{x}^L)$, and the observed image as $f_O(\vec{x}^O)$, where $\vec{x}^S$ is the coordinate in the source plane, and $\vec{x}^L$ and $\vec{x}^O$ are the positions in the image plane. We have the following relations:
\begin{eqnarray}
&&f_L(\vec{x}^L)=f_S(\vec{x}^S), \quad\quad \vec{x}^S={\mathbf M}\vec{x}^L, \nonumber \\
&&f_O(\vec{x}^O)=\int d^2\vec{x}^L W(\vec{x}^O-\vec{x}^L)f_L(\vec{x}^L),
\label{define1}
\end{eqnarray} 
where $W$ is the PSF, ${\mathbf M}$ is the lensing distortion matrix typically defined as: ${\mathbf M}_{ij}=\delta_{ij}-\phi_{ij}$ with $\phi_{ij}=\delta_{ij}-\partial x^S_i/\partial x^L_j$ being the spatial derivatives of the lensing deflection angle. $\phi_{ij}$ is often replaced by the convergence $\kappa$ [$=(\phi_{11}+\phi_{22})/2$] and the two shear components $\gamma_1$ [$=(\phi_{11}-\phi_{22})/2$] and $\gamma_2$ ($=\phi_{12}$). Note that here we adopt the convention used in most other lensing literatures (\eg, \cite{kaiser95,bs01}) to define the shear components and the convergence. They are different from the definitions in the other papers of this series \cite{zhang08,zhang09,zhang10} at the second order. Their differences are shown and discussed in the appendix. 

Z08 has shown that the following relations can be used to measure the cosmic shear:
\begin{eqnarray}
\label{shear12PSF}
\frac{1}{2}\frac{\langle\langle (\partial_1f_O)^2-(\partial_2f_O)^2\rangle_g\rangle_{en}}{\langle\langle (\partial_1f_O)^2+(\partial_2f_O)^2+\Delta\rangle_g\rangle_{en}}&=&-\gamma_1, \nonumber \\
\frac{\langle\langle\partial_1f_O\partial_2f_O\rangle_g\rangle_{en}}{\langle\langle (\partial_1f_O)^2+(\partial_2f_O)^2+\Delta\rangle_g\rangle_{en}}&=&-\gamma_2,
\end{eqnarray}
where $\partial_i$ denotes $\partial /\partial x_i$, and
\begin{equation}
\label{Delta}
\Delta=\frac{\beta^2}{2}\vec{\nabla}f_O\cdot\vec{\nabla}(\nabla^2f_O).
\end{equation}
$\beta$ is the scale radius of the isotropic Gaussian PSF $W_{\beta}$ defined as:
\begin{equation}
\label{igaussian}
W_{\beta}(\vec{\theta})=\frac{1}{2\pi\beta^2}\exp\left(-\frac{\vert\vec{\theta}\vert^2}{2\beta^2}\right).
\end{equation} 
$\langle\cdots\rangle_g$ means taking the average over the surface brightness field of a single galaxy. $\langle\cdots\rangle_{en}$ means taking the average over an ensemble of galaxies. Note that whatever the original PSF is, it is always transformed into the desired isotropic Gaussian form that is defined in eq.(\ref{igaussian}). This is allowed as long as the scale radius $\beta$ of the target PSF is somewhat larger than that of the original PSF. With the help of Fourier transformation, we now show how to improve the accuracy of eq.(\ref{shear12PSF}) to the second order in shear/convergence. Readers who are not interested in the mathematical derivations can jump to eq.(\ref{mmain}) and eq.(\ref{mmain_rs}) for the main results.

For a general 2D field $f(\vec{x})$, let us use $\widetilde{f}(\vec{k})$ to denote its Fourier transformation, which is defined as:
\begin{equation}
\label{FT}
\widetilde{f}(\vec{k})=\int d^2\vec{x} e^{i\vec{k}\cdot\vec{x}}f(\vec{x}).
\end{equation}
The terms in eq.(\ref{shear12PSF}) can be written as integrations of the Fourier modes of the images weighted by proper functions of the wave vector as follows \cite{zhang10}:
\begin{eqnarray}
\label{RF_derive}
\left\langle\partial_if_O(\vec{x})\partial_jf_O(\vec{x})\right\rangle_g&=&\frac{1}{S}\frac{\int d^2\vec{k}}{(2\pi)^2}k_ik_j\left\vert\widetilde{f}_O(\vec{k})\right\vert^2,\nonumber \\
\left\langle\vec{\nabla}f_O\cdot\vec{\nabla}(\nabla^2f_O)\right\rangle_g&=&-\frac{1}{S}\frac{\int d^2\vec{k}}{(2\pi)^2}\left\vert\vec{k}\right\vert^4\left\vert\widetilde{f}_O(\vec{k})\right\vert^2,
\end{eqnarray}
where $S$ is the total area of the map. Consequently, eq.(\ref{shear12PSF}) can be written in Fourier space as:
\begin{eqnarray}
\label{shearFourier}
\frac{1}{2}\frac{\left\langle P_{20}-P_{02}\right\rangle_{en}}{\left\langle P_{20}+P_{02}-\beta^2D_4/2\right\rangle_{en}}&=&-\gamma_1,\nonumber \\
\frac{\left\langle P_{11}\right\rangle_{en}}{\left\langle P_{20}+P_{02}-\beta^2D_4/2\right\rangle_{en}}&=&-\gamma_2,
\end{eqnarray}
where 
\begin{eqnarray}
\label{defineP}
P_{ij}&=&\int d^2\vec{k}k_1^ik_2^j\left\vert\widetilde{f_O}(\vec{k})\right\vert^2,\nonumber \\
D_n&=&\int d^2\vec{k}\left\vert\vec{k}\right\vert^n\left\vert\widetilde{f_O}(\vec{k})\right\vert^2.
\end{eqnarray}
Note that $D_4=P_{40}+2P_{22}+P_{04}$. To find out how $P_{ij}$ changes under lensing to the second order in shear/convergence, we use the following relations that are derived from eq.(\ref{define1}) and properties of Fourier transformation:
\begin{eqnarray}
\label{FT2}
\widetilde{f_L}(\vec{k}^L)&=&\int d^2\vts \left\vert \mathrm{det} \left(\frac{\partial \vtl}{\partial \vts}\right)\right\vert e^{i\vec{k}^L\cdot(\mathbf{M}^{-1}\vts)}f_S(\vts)\nonumber \\
&=&\vert \mathrm{det} (\mathbf{M}^{-1})\vert\int d^2\vts e^{i(\mathbf{M}^{-1}\vec{k}^L)\cdot\vts}f_S(\vts)\nonumber \\
&=&\vert \mathrm{det} (\mathbf{M}^{-1})\vert\widetilde{f_S}(\mathbf{M}^{-1}\vec{k}^L),\nonumber \\
\nonumber \\
\widetilde{f_O}(\vec{k})&=&\widetilde{W}_{\beta}(\vec{k})\widetilde{f_L}(\vec{k}).
\end{eqnarray}
Note that in the above equations, we have assumed that the PSF is isotropic Gaussian ($W_{\beta}$), whose Fourier transformation is $\widetilde{W}_{\beta}(\vec{k})$ [$=\exp(-\beta^2\vert\vec{k}\vert^2/2)$]. The dependence of $P_{ij}$ on the cosmic shear can be derived as follows: 
\begin{eqnarray}
\label{pldl}
P_{ij}&=&\int d^2\vec{k}k_1^ik_2^j\left\vert\widetilde{W}_{\beta}(\vec{k})\vert \mathrm{det} (\mathbf{M}^{-1})\vert\widetilde{f_S}(\mathbf{M}^{-1}\vec{k})\right\vert^2\nonumber \\
&=&\vert\mathrm{det} (\mathbf{M}^{-1})\vert\int d^2\vec{k}(\mathbf{M}\vec{k})_1^i(\mathbf{M}\vec{k})_2^j\left\vert\widetilde{W}_{\beta}(\mathbf{M}\vec{k})\widetilde{f_S}(\vec{k})\right\vert^2.
\end{eqnarray}
The last step is achieved by re-defining $\mathbf{M}^{-1}\vec{k}$ as $\vec{k}$. Expanding eq.(\ref{pldl}) up to the second order in shear/convergence, and using the fact that the intrinsic galaxy shapes are isotropic (\ie, $\langle \left\vert\widetilde{f}_S(\vec{k})\right\vert^2\rangle_{en}$ only depends on $\left\vert\vec{k}\right\vert$), it is now straightforward, though a little tedious, to show the following:
\begin{eqnarray}
\label{main_result}
\left\langle P_{20}-P_{02}\right\rangle_{en}&=&-2\gamma_1(1+\kappa)\left\langle D_2^S\right\rangle_{en}+\gamma_1(1-5\kappa)\beta^2\left\langle D_4^S\right\rangle_{en}\nonumber \\
&&+2\gamma_1\kappa\beta^4\left\langle D_6^S\right\rangle_{en}+O(\gamma^3),\nonumber \\
\nonumber \\
\left\langle P_{11}\right\rangle_{en}&=&-\gamma_2(1+\kappa)\left\langle D_2^S\right\rangle_{en}+\frac{1}{2}\gamma_2(1-5\kappa)\beta^2\left\langle D_4^S\right\rangle_{en}\nonumber \\
&&+\gamma_2\kappa\beta^4\left\langle D_6^S\right\rangle_{en}+O(\gamma^3),\nonumber \\
\nonumber \\
\left\langle P_{20}+P_{02}\right\rangle_{en}&=&(1+2\gamma_1^2+2\gamma_2^2)\left\langle D_2^S\right\rangle_{en}\nonumber \\
&&+(2\kappa-\kappa^2-3\gamma_1^2-3\gamma_2^2)\beta^2\left\langle D_4^S\right\rangle_{en}\nonumber \\
&&+(2\kappa^2+\gamma_1^2+\gamma_2^2)\beta^4\left\langle D_6^S\right\rangle_{en}+O(\gamma^3),\nonumber \\
\nonumber \\
\left\langle D_4\right\rangle_{en}&=&(1-2\kappa+\kappa^2+5\gamma_1^2+5\gamma_2^2)\left\langle D_4^S\right\rangle_{en}\nonumber \\
&&+(2\kappa-5\kappa^2-5\gamma_1^2-5\gamma_2^2)\beta^2\left\langle D_6^S\right\rangle_{en}\nonumber \\
&&+(2\kappa^2+\gamma_1^2+\gamma_2^2)\beta^4\left\langle D_8^S\right\rangle_{en}+O(\gamma^3),
\end{eqnarray} 
where
\begin{equation}
\label{defineDS}
D_n^S=\int d^2\vec{k}\left\vert\vec{k}\right\vert^n\left\vert\widetilde{W}_{\beta}(\vec{k})\widetilde{f_S}(\vec{k})\right\vert^2.
\end{equation}
Note that $O(\gamma^3)$ refers to terms of the third or higher orders in $\gamma_{1,2}$, $\kappa$. From eq.(\ref{main_result}), we find the main results of this paper:
\begin{eqnarray}
\label{mmain}
\frac{1}{2}\frac{\left\langle P_{20}-P_{02}\right\rangle_{en}}{\left\langle P_{20}+P_{02}-\beta^2D_4/2\right\rangle_{en}}&=&-\frac{\gamma_1}{1-\kappa}+O(\gamma^3),\nonumber \\
\frac{\left\langle P_{11}\right\rangle_{en}}{\left\langle P_{20}+P_{02}-\beta^2D_4/2\right\rangle_{en}}&=&-\frac{\gamma_2}{1-\kappa}+O(\gamma^3).
\end{eqnarray}
The terms on the right sides of eq.(\ref{mmain}) are called reduced shears. For convenience, in the rest of the paper, we use $g_{1,2}$ to represent $\gamma_{1,2}/(1-\kappa)$ respectively. One can equivalently write down the above formula in real space as:
\begin{eqnarray}
\label{mmain_rs}
\frac{1}{2}\frac{\langle\langle (\partial_1f_O)^2-(\partial_2f_O)^2\rangle_g\rangle_{en}}{\langle\langle (\partial_1f_O)^2+(\partial_2f_O)^2+\Delta\rangle_g\rangle_{en}}&=&-g_1+O(\gamma^3), \nonumber \\
\frac{\langle\langle\partial_1f_O\partial_2f_O\rangle_g\rangle_{en}}{\langle\langle (\partial_1f_O)^2+(\partial_2f_O)^2+\Delta\rangle_g\rangle_{en}}&=&-g_2+O(\gamma^3).
\end{eqnarray}
We have just shown that to the second order in accuracy, our shear estimators yield exactly the reduced shears. It is important to note that the ensemble averages in eq.(\ref{mmain}) / eq.(\ref{mmain_rs}) should be taken before the ratios. This point is extensively discussed in \cite{zhang10}.

The novelty of our result may be questioned by those who believe that the existing shear measurement methods already claim to measure the reduced shears. Such a judgement is mainly based on the fact that the lensing distortion matrix ${\mathbf M}$ can be re-written as:
\begin{equation}
\label{re_define_M}
{\mathbf M}=(1-\kappa)\left(\begin{array}{cc}
1-g_1 & -g_2 \\
-g_2 & 1+g_1 \end{array}\right),
\end{equation}
which shows that $\kappa$ changes the galaxy size, and the shape distortions are exclusively due to the reduced shears $g_{1,2}$. Therefore, the results of shape/shear measurements should only depend on the reduced shears, not the convergence, as long as the shape measurement is decoupled from the galaxy size. This is usually the case in the absence of PSF. When a PSF is present, however, the convergence $\kappa$ can modify the observed galaxy shape through changing the galaxy size relative to that of the PSF. Generally, for spin-2 shear estimators\footnote{According to \cite{zhang10}, we only need to consider spin-2 shear estimators.} such as those in most other existing methods, the results can generally be written as power series of $g_{1,2}$ and $\kappa$ as\footnote{The formula is initiated by the referee of this paper.}:
\begin{equation} 
\label{expansion}
g_{1,2}\left[a+b\kappa+c(g_1^2+g_2^2)+\cdots\right], 
\end{equation}
in which $a$, $b$, and $c$ are coeficients that may require calibrations ($a=1$ and $b=0$ in our method). Obviously, unless $b$ is proved to be zero, a shear measurement method does not directly measure the reduced shear at the second order level. For example, if we mistakenly use $(1/2)\langle P_{20}-P_{02}\rangle_{en}/\langle P_{20}+P_{02}\rangle_{en}$ and $\langle P_{11}\rangle_{en}/\langle P_{20}+P_{02}\rangle_{en}$ defined in eq.(\ref{main_result}) as shear estimators (\ie, neglecting the term $D_4$ for the PSF correction), we find the following formulae to the second order in accuracy:
\begin{equation}
\label{mistakenly}
\frac{1}{2}\frac{\left\langle P_{20}-P_{02}\right\rangle_{en}}{\left\langle P_{20}+P_{02}\right\rangle_{en}}=-g_1(a+b\kappa), \quad\quad \frac{\left\langle P_{11}\right\rangle_{en}}{\left\langle P_{20}+P_{02}\right\rangle_{en}}=-g_2(a+b\kappa),
\end{equation}
in which 
\begin{equation}
\label{definition_ab}
a=1-\frac{u}{2}, \quad\quad b=u+u^2-\beta^4\frac{\left\langle D_6^S\right\rangle_{en}}{\left\langle D_2^S\right\rangle_{en}}, \quad\quad u=\beta^2\frac{\left\langle D_4^S\right\rangle_{en}}{\left\langle D_2^S\right\rangle_{en}}.
\end{equation}
This simple example has shown that if accuracy at the second order is desired, one should pay attention to the possible existence of terms proportional to $g_{1,2}\kappa$ in addition to the reduced shears. If one plans to achieve the second order accuracy by calibrating the multiplicative factor $a+b\kappa$ using numerical simulations, it is important to realize that the factor typically depends on $\kappa$. These troubles are not present in our method.

\section{A Useful Theorem}
\label{theorem}
 
Since the pre-lensing galaxy sizes are not known a priori, it would be ideal for the statistical expectation of a shear estimator to depend only on the reduced shears, \ie, to be completely decoupled from the galaxy sizes, or the convergence $\kappa$. In terms of eq.(\ref{expansion}), this corresponds to $b=0$, and the coefficients in front of other high order $\kappa$ terms are all zero. This is, however, highly nontrivial to achieve in the presence of PSF, especially when all high order shear/convergence terms are considered. Less ambitiously, we have shown in the previous section that the shear estimators of Z08 indeed decouple from the galaxy sizes up to the second order in shear/convergence. Our calculation is quite laborious, and the situation is likely similar in all other shear measurement methods. Fortunately, we find a simple way of knowing the property of a shear estimator at the second order according to its first order results:
\begin{theorem}
\label{reduced_shear}
For spin-2 shear estimators whose statistical expectation values (in the presence of PSF, but without noise) can be written in the form of eq.(\ref{expansion}), if $a\equiv 1$ (\ie, the method is free of calibrations at the first order), we must have $b\equiv 0$.
\end{theorem}
\begin{proof}
If $a\equiv 1$, the statistical expectations of the shear estimators can be written as:
\begin{equation} 
\label{expansion2}
g_{1,2}\left[1+b\kappa+c(g_1^2+g_2^2)+\cdots\right], 
\end{equation}
For a given set of observed galaxies, let us consider the following two cases with fixed $g_1$ and $g_2$:

1. $\kappa=0$;

2. $\kappa\neq 0$, and the galaxies are intrinsically larger than those in the first case by a factor of $(1-\kappa)$. 

Since the observed galaxy images in the two cases are identical, the shear measurement method should yield the same results. On the other hand, according to eq.(\ref{expansion2}), we should get (if accurate to the second order) $g_{1,2}$ in the first case, and $g_{1,2}(1+b\kappa)$ in the second case. Therefore, we have $b\equiv 0$. \qed
\end{proof}

The above theorem shows that $a=1$ is a sufficient (not necessary) condition for $b=0$. For $a$ being any nonzero constant, the shear estimator can be trivially rescaled to make $a=1$, which again leads to $b=0$. When $a$ is a function of the galaxy morphology and size in a method, the first case in the above proof yields $g_{1,2}a'$, while the second case gives $g_{1,2}(a+b\kappa)$. Since $a$ and $a'$ can naturally have a difference of order $\kappa$ due to the change of the intrinsic galaxy sizes, the value of $b$ is not necessarily equal to zero. 
 
The theorem can be phrased in plain words as: if a method measures exactly the shears at the first order without the need of calibrations, it must yield exactly the reduced shears at the second order level. One can immediately apply this theorem to the method of Z08 to show why it accurately measures the reduced shear without the lengthy calculation in \S\ref{method}. The theorem is equally useful for judging the accuracy of any shear measurement method at the second order based on its properties at the first order.

\section{Numerical Test}
\label{test}

In this section, we test the accuracy of shear recovery with a large ensemble ($\gtrsim 10^7$) of mock galaxies. In principle, eq.(\ref{mmain}) / eq.(\ref{mmain_rs}) allows us to recover the cosmic shear to a sub-percent level accuracy. In practice, however, one needs to consider at least two important factors: the pixelation effect and the photon noise. The purpose of this section is to test this method under realistic conditions, and to show how these factors may affect the accuracy of our method.

\cite{zhang09} has introduced useful interpolation methods to treat the pixelation effect, which becomes a problem for shear measurement when the CCD pixel size is comparable to the size of the PSF. We find that these methods are not accurate at the sub-percent level in terms of shear recovery. On the other hand, in the course of this work, we surprisingly find that if it is only for the purpose of shear measurement, oversampling galaxy images does not seem necessary. At least for the PSFs used in this work, the shear estimators calculated in Fourier space converge quickly when the pixel size is less than about $1/3$ of the FWHM of the PSF. This will be demonstrated in \S\ref{pixelation}. 

The treatment of the background photon noise has also been discussed in \cite{zhang09}. It can in principle be used to deal with both the background fluctuation and its Poisson noise, since they are both independent of the source flux. There is no need for any modifications of the treatment, because the quantities that we need to measure from each galaxy are the same as those discussed in Z08. The treatment is indeed simple: for each galaxy, one subtracts from the nominators and denominators in eq.(\ref{mmain}) / eq.(\ref{mmain_rs}) the contributions from the photon noise that are estimated from a neighboring map of pure noise \cite{zhang09}. Note again that the ratios should be taken after the ensemble averages. As will be shown in \S\ref{ssNoise}, the shear measurement errors due to the background noise can be removed cleanly with the treatment of \cite{zhang09}. On the other hand, we do not yet have a way to correct the shear measurement errors due to the source Poisson noise. In this paper, we simply perform the shear measurements without additional treatment of the source Poisson noise. A further development in this aspect will be studied in a separate work.

\subsection{General Setup}
\label{general_setup}

Each of our mock galaxies is placed at the center of a $192\times192$ grid. The grid size is used as the length unit in the rest of this paper. The PSF has a truncated Moffat profile used in the GREAT08 project (see \cite{bridle09b} for details): 
\begin{eqnarray}
\label{PSFs}
W_{I}(r)\propto\begin{cases}\left(1+\frac{r}{r_d}\right)^{-3.5}, & r<r_c \\ 0, & r\geq r_c\end{cases} 
\end{eqnarray}
The FWHM of this PSF is very close to $r_d$. We always set $r_c=3r_d$ in the simulations of this paper. In our shear measurement method, the PSF is always transformed into the isotropic Gaussian form in Fourier space before the shear measurement is carried out. The scale radius ($\beta$) of the target PSF defined in eq.(\ref{igaussian}) is set to $0.7r_d$, so that the FWHM of the target PSF is slightly larger than that of the original PSF. $r_d$ is set to $12$ in unit of the grid size. Note that the grid size is not equivalent to the pixel size. The later is always chosen to be an integer multiple of the grid size in this paper. 

Our mock galaxies are made of point sources that are generated by 2D random walks. There are at least four main purposes for this arrangement: 

{\bf 1.} To maximize to some extent the richness of galaxy morphologies;

{\bf 2.} The lensing effect can be exactly mimicked by simply changing the positions of the point sources; 

{\bf 3.} Convolution of the galaxy image with any PSF is trivial and easy;

{\bf 4.} It is extremely fast to generate such galaxies.

Based on the above facts, we encourage everyone working in the field of shear measurement to test their methods with the random-walk-generated mock galaxies\footnote{We are aware of the fact that our random-walk galaxies do not have certain properties of real galaxy, such as the sharp cusp, the long tails of the de-Vaucouleurs profile, or any systematic trends in ellipticity gradients with radius. However, sharp cusp should not cause any singularity in the method because of the smoothing by the PSF. Neither the long tails of the de-Vaucouleurs profiles nor any systematic trends in ellipticity gradients with radius seems to be able to affect the accuracy of this method as long as the whole/complete galaxy image is captured/used. A more detailed study of these issues will be included in a future work.}. In our simulations, each such mock galaxy is made of forty point sources, whose positions are determined by forty steps of 2D random walk. Each step size is a random number between $0$ and $1$. The direction of every step is completely random in 2D. For each galaxy, the first step starts from the center of the grid. The ending position of the $i^{th}$ step is where the $i^{th}$ point source is. In the course of the random walk, if the distance of the $i^{th}$ point to the center of the grid is more than $6$, we restart the $i^{th}$ step from the center of the map. Every point source of a galaxy is assumed to carry the same luminosity. The resulting post-seeing galaxy has roughly the same FWHM as that of the PSF.  

To test the shear recovery accuracy of our method, we use six sets of input shear values ($\gamma_1$, $\gamma_2$): ($0.05$, $-0.05$), ($0.03$, $-0.03$), ($0.01$, $-0.01$), ($-0.01$, $0.01$), ($-0.03$, $0.03$), ($-0.05$, $0.05$). $\kappa$ is fixed at $0.05$. For each set of the shear values, we use about $10^7$ mock galaxies to recover the shear. To calibrate the shear recovery accuracy quantitatively, we adopt the commonly used multiplicative bias $m$ and additive bias $c$ defined as follows: 
\begin{eqnarray}
\label{mc}
g_1^{measured}&=&(1+m_1)g_1^{input}+c_1,\nonumber \\
\nonumber \\
g_2^{measured}&=&(1+m_2)g_2^{input}+c_2.
\end{eqnarray}
For each shear component, we calculate the two bias parameters by fitting a linear relation between the measured and the input values of the reduced shear.

\subsection{Test the Pixelation Effect}
\label{pixelation}

We set the pixel size to be a multiple of the grid size to mimic the pixelation effect. Since we need to use Fast Fourier Transform (FFT) in our method, the size of each galaxy postage stamp is chosen to be an integer power of $2$ times the pixel size along both sides. The power integer is chosen to be the value that allows the postage stamp to cover the most area of the whole grid. For example, if the ratio of the pixel size to the grid size is set to $5$, the postage stamp contains $32\times 32$ pixels, because $5\times 32 \leq 192$, and $5\times 64 > 192$. Note that our choice of the postage stamp size is large enough for avoiding cutting off the edges of the simulated galaxies. The simulations in this section do not contain noise. 

Under-sampling of the galaxy images can certainly affect the shear recovery accuracy. We have discussed a few interpolation methods in \cite{zhang09} to help reduce the systematic errors due to the pixelation effect. The most accurate methods are found to be the so called ``Log-Bicubic'' or ``Log-Spline'' methods (equally good), which simply refer to performing the traditional ``Bicubic'' and ``Spline'' methods on the logarithm of the data instead of the data itself. The focus of the previous work is on the Gaussian PSF. The PSF defined in eq.(\ref{PSFs}) has a power-law decaying form when the distance to the center is large. For this type of PSF, it is perhaps not surprising that the Log-Bicubic method is found not to work as well as in the Gaussian PSF case, although it is still better than the traditional Bicubic method. On the other hand, fortunately, we find that {\it shear measurement in Fourier space is quite robust even for marginally under-sampled images} (without the need of interpolation), regardless of whether the PSF is a Gaussian function or a power-law form. This is demonstrated in fig.\ref{pixelation_GREAT08} and fig.\ref{pixelation_Gaussian}, which shows how fast the measured value of the first shear component from a single galaxy (defined on the left side of the first half of eq.(\ref{mmain})) converges for an increasingly smaller pixel size. The results are shown for three methods: using the Bicubic interpolation method (the dotted curves); using the Log-Bicubic method (the dashed curves); direct Fourier transformation/shear measurement (the solid curves). Fig.\ref{pixelation_GREAT08} and fig.\ref{pixelation_Gaussian} are for the PSF defined in eq.(\ref{PSFs}) and the Gaussian PSF respectively. The two PSFs have the same FWHM. Clearly, according to the figures, the direct measurement (\ie, no interpolation or any other treatment for the pixelation effect) consistently yields no more than $0.5\%$ relative deviations from the correct answers as long as the pixel size is not larger than $4$, which corresponds to about $1/3$ of the FWHM of the PSF in both cases. It is useful to note that the ratio $1/3$ between the pixel size and the FWHM of the PSF is indeed reasonable in practice, and used in the GREAT08 project. The Log-Bicubic method, on the other hand, can cause a few percent error on the shear measurement when the PSF has a power-law decaying form, despite its superiority in the Gaussian PSF case. 

\begin{figure}
\centerline{\epsfxsize=9cm\epsffile{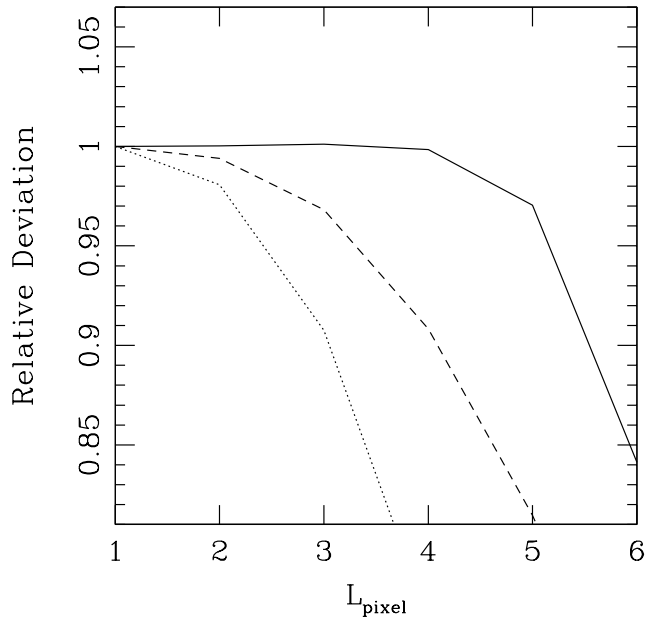}}
\caption{This figure compares the qualities of three different ways of treating the pixelation effect. The input PSF is defined in eq.(\ref{PSFs}), with FWHM$\approx r_d=12$. For each chosen pixel size $L_{pixel}$ (in unit of the grid size), we plot the measured first component of the reduced shear from a single galaxy (defined on the left side of the first half of eq.(\ref{mmain})) normalized by its measured value at $L_{pixel}=1$. The solid curve is from a direct Fourier transformation without any interpolation; the dashed curve is from the data interpolated by the Log-Bicubic method; the dotted line is from the data interpolated by the Bicubic method. In the second and last cases, the pixel size of the interpolated images is always set to $1$.}
\label{pixelation_GREAT08}
\end{figure}

\begin{figure}
\centerline{\epsfxsize=9cm\epsffile{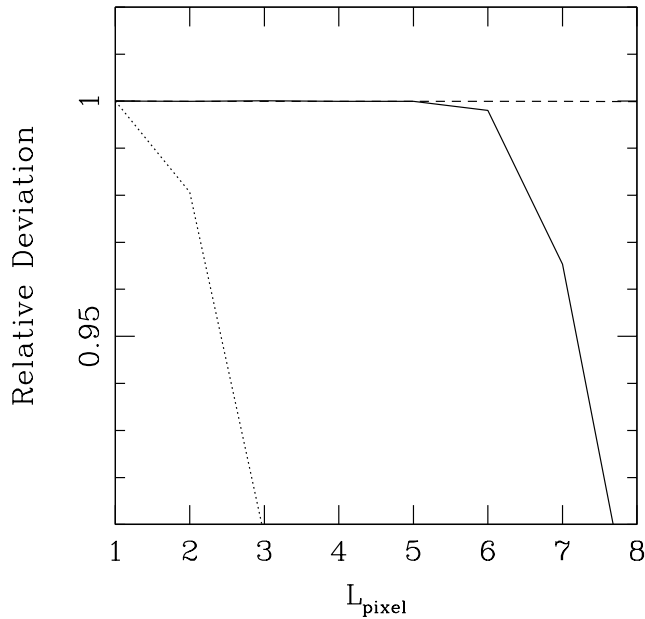}}
\caption{Same as fig.\ref{pixelation_GREAT08}, except that the input PSF has a Gaussian form define in eq.(\ref{igaussian}). The FWHMs  of the two PSFs are equal.}
\label{pixelation_Gaussian}
\end{figure}

For further checking the accuracy of direct Fourier space measurements on pixelated images,  we use a large ensemble of mock galaxies and the input shear components given in the previous section to study how the values of the multiplicative and additive bias vary with the pixel size. Here and in the rest of the paper, we use the PSF defined in eq.(\ref{PSFs}). In table \ref{pixelation_cases1} and \ref{pixelation_cases2}, we show the results for six different choices of the pixel size. For each pixel size, we consider four choices of galaxy size, which are listed in the tables as small, medium, large, largest galaxies (written as ``Gal.'' in the tables for abbreviation). The generation of the medium size galaxies are described in \S\ref{general_setup}. The small, large, largest galaxies are generated through the same procedures, except that the overall scales of the galaxies are multiplied by factors of 0.5, 2, and 4 respectively before they are convolved with the PSF. The average FWHMs of the post-seeing galaxies of small, medium, large, and very large sizes are equal to 1.02, 1.05, 1.16, and 1.43 times the FWHM of the PSF respectively. In our simulations, when a galaxy is generated, it is used for different shear values and pixel sizes for the purpose of saving time. Galaxies of different sizes are always generated using different sets of random seeds instead of simple rescaling. $6.5\times 10^6$ mock galaxies are used to measure the biases in every case of the tables. 

According to table \ref{pixelation_cases1} and \ref{pixelation_cases2}, for all cases with $L_{pixel}\leq 4$, the multiplicative and additive biases are consistent with zero at the sub-percent accuracy level . This agrees with the conclusions from fig.\ref{pixelation_GREAT08} and \ref{pixelation_Gaussian}. Note that there are roughly $0.1\%$ residual systematic errors, which are likely from third order corrections. On the other hand, the data on the tables show that the pixelation effect becomes less important for larger galaxies. This is consistent with our intuitions. In the next section, we include noise in the simulations to test our method under realistic conditions.

\begin{table}
\centering
\begin{tabular}{ccccc}

\hline\hline		
             &  Small Gal.             & Medium Gal.   &  Large Gal.  & Largest Gal.\\
\hline\hline		                                                             
$L_{pixel}=1$ & $m_1(10^{-3}):2.0\pm 1.5$ & $-0.8\pm 1.5$  & $3.8\pm 1.5$  & $2.8\pm 1.4$\\
             & $c_1(10^{-5}):9.6\pm 4.8$ & $4.3\pm 4.8$   & $-2.1\pm 4.7$ & $-4.8\pm 4.7$\\
\hline
$L_{pixel}=2$ & $3.2\pm 1.5$             & $-0.5\pm 1.5$  & $3.9\pm 1.5$  & $2.9\pm 1.4$\\
             & $9.4\pm 4.8$             & $4.2\pm 4.8$   & $-2.1\pm 4.7$ & $-4.8\pm 4.7$\\
\hline
$L_{pixel}=3$ & $7.2\pm 1.5$             & $0.6\pm 1.5$   & $4.2\pm 1.5$  & $3.0\pm 1.4$\\
             & $10.2\pm 4.8$            & $4.3\pm 4.8$   & $-2.1\pm 4.7$ & $-4.8\pm 4.7$\\
\hline
$L_{pixel}=4$ & $-2.5\pm 1.5$            & $-1.9\pm 1.5$  & $3.5\pm 1.5$  & $2.8\pm 1.4$\\
             & $8.9\pm 4.8$             & $4.1\pm 4.8$   & $-2.1\pm 4.7$ & $-4.8\pm 4.7$\\
\hline
$L_{pixel}=5$ & $-91.2\pm 1.3$           & $-30.6\pm 1.4$ & $-4.6\pm 1.4$ & $0.4\pm 1.4$\\
             & $9.8\pm 4.3$             & $4.1\pm 4.6$   & $-2.1\pm 4.7$ & $-4.8\pm 4.6$\\
\hline
$L_{pixel}=6$ & $-372\pm 1$              & $-174.5\pm 1.2$& $-58.4\pm 1.4$& $-16.2\pm 1.4$\\
             & $7.5\pm 3.1$             & $3.2\pm 3.9$   & $-2.0\pm 4.4$ & $-4.7\pm 4.6$\\

\hline\hline		

\end{tabular}
\caption{The multiplicative bias $m_1$ and the additive bias $c_1$ of the first component of the reduced shear measured using the method of this paper. In each data cell, the upper value is $m_1$ in unit of $10^{-3}$, and the lower value is $c_1$ in unit of $10^{-5}$. The results are shown for six choices of the pixel size and four choices of the galaxy size. Noise is not included in this set of simulations. }
\label{pixelation_cases1}
\end{table}

\begin{table}
\centering
\begin{tabular}{ccccc}

\hline\hline		
             &  Small Gal.              & Medium Gal.    &  Large Gal.   & Largest Gal.\\
\hline\hline		                                                             
$L_{pixel}=1$ & $m_2(10^{-3}):-0.9\pm 1.5$ & $-1.6\pm 1.5$   & $-0.4\pm 1.5$  & $1.6\pm 1.4$\\
             & $c_2(10^{-5}):7.3\pm 4.8$  & $-0.5\pm 4.8$   & $7.4\pm 4.7$   & $1.2\pm 4.7$\\
\hline
$L_{pixel}=2$ & $0.3\pm 1.5$              & $-1.3\pm 1.5$   & $-0.3\pm 1.5$  & $1.7\pm 1.4$\\
             & $7.6\pm 4.8$              & $-0.4\pm 4.8$   & $7.5\pm 4.7$   & $1.2\pm 4.7$\\
\hline
$L_{pixel}=3$ & $4.3\pm 1.5$              & $-0.2\pm 1.5$   & $0.0\pm 1.5$   & $1.8\pm 1.4$\\
             & $8.9\pm 4.8$              & $-0.3\pm 4.8$   & $7.5\pm 4.7$   & $1.2\pm 4.7$\\
\hline
$L_{pixel}=4$ & $-5.4\pm 1.5$             & $-2.7\pm 1.5$   & $-0.7\pm 1.5$  & $1.5\pm 1.4$\\
             & $10.2\pm 4.8$             & $-0.3\pm 4.8$   & $7.4\pm 4.7$   & $1.2\pm 4.7$\\
\hline
$L_{pixel}=5$ & $-100.3\pm 1.3$           & $-32.2\pm 1.4$  & $-8.5\pm 1.4$  & $-0.5\pm 1.4$\\
             & $11.0\pm 4.3$             & $-0.1\pm 4.6$   & $7.4\pm 4.7$   & $1.2\pm 4.6$\\
\hline
$L_{pixel}=6$ & $-420.2\pm 0.9$           & $-192.3\pm 1.2$ & $-62.8\pm 1.4$ & $-16.8\pm 1.4$\\
             & $54.2\pm 2.8$             & $4.4\pm 3.9$    & $7.2\pm 4.4$   & $1.2\pm 4.6$\\

\hline\hline		

\end{tabular}
\caption{Same as table \ref{pixelation_cases1}, except that it is for multiplicative and additive biases ($m_2$ and $c_2$) of the second component of the reduced shear.}
\label{pixelation_cases2}
\end{table}

\subsection{Tests with Photon Noise}
\label{ssNoise}

There are three types of photon noise: 1. the spatial fluctuations of the sky background; 2. The Poisson photon counting noise of the background photons; 3. The source Poisson noise. The noise treatment defined in \cite{zhang09} only deals with the background noise. We currently do not have a way to treat the source Poisson noise, which can certainly introduce systematic errors to the shear recovery. In the weak lensing community,  one typically use the signal-to-noise-ratio (SNR) to denote the noise amplitude relative to the source signal within the half-light radius of the source (post-seeing). Here, we point out that at least within the scope of our method, {\it it is important to differentiate between background and source noise} for the reasons just discussed. In other words, for any given SNR, whether most of the noise is due to the background or the source can significantly affect the shear recovery accuracy. The ratio of the background noise to the source noise is determined by the luminosities of the source and the background.

The purpose of this section is to study the shear recovery accuracy under different SNR. For a given SNR, we consider three cases: 1. solely background noise; 2. the background and source Poisson noises are even; 3. solely source Poisson noise. The background noise in our simulations is generated as the Poisson noise of a homogeneous background. To save computational time, the background fluctuations are not included in this paper. This should not affect our main conclusions. Fig.\ref{gal_cases} shows sample images of a single simulated galaxy with three SNR values: 10, 20, 40. For each SNR, the three cases just discussed are shown in the same row.  

\begin{figure}
\centerline{\epsfxsize=9cm\epsffile{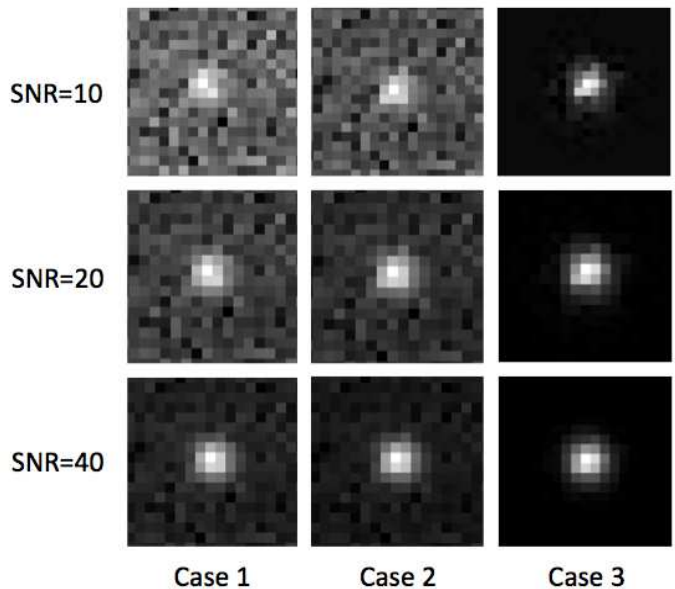}}
\caption{Sample images of the same galaxy under 9 different noise conditions. Images in the same row have the same SNR, which is shown on the left. Images in each column share one of the three cases: 1. solely Poisson noise of the background; 2. even contributions from the Poisson noises of the background and the source; 3. solely source Poisson noise. }
\label{gal_cases}
\end{figure}

The galaxies and PSF used in the simulations of this section are discussed in \S\ref{general_setup}. The pixel size $L_{pixel}$ is taken to be $4$, which matches what is assumed in the GREAT08 project. For treating the background noise using the method of \cite{zhang09}, we randomly generate an image of pure background noise for each source image. Similar to what is done in \S\ref{pixelation}, when a galaxy is generated, it is used for different shear values and noise levels for the purpose of saving time. For each simulated galaxy, the background Poisson noise in each pixel is generated as a Gaussian random number with a unitary variance\footnote{The Gaussian distribution is a good approximation to the Poisson distribution when the background photon number in each pixel is not a very small number.}, and rescaled according to the assumed SNR due to the background noise (${\rm SNR}_B$). A separate image of pure background noise is generated in the same way with different random seeds.  The source Poisson noise in each pixel is generated as a Gaussian random number with unitary variances multiplied by the square root of the source flux of the same pixel. It is then rescaled according to the value of SNR due to the source Poisson noise (${\rm SNR}_S$). Since the source Poisson noise and the background noise are not correlated, we use the following equation to relate SNR, ${\rm SNR}_B$, ${\rm SNR}_S$:
\begin{equation}
\label{SNR}
\frac{1}{{\rm SNR}^2}=\frac{1}{{\rm SNR}_B^2}+\frac{1}{{\rm SNR}_S^2}
\end{equation}
For each given SNR, case 1, 2, and 3 correspond to ${\rm SNR}_B/{\rm SNR}_S=0$, $1$, and $\infty$ respectively. We consider six choices of SNR: 10, 20, 30, 40, 60, 80. The main results are shown in table \ref{noise_cases1} and \ref{noise_cases2}. $2.6\times 10^7$ mock galaxies are used to measure the biases in every case of the tables. 

\begin{table}
\centering
\begin{tabular}{cccc}

\hline\hline		
                     &  Case 1                                    & Case 2                   &   Case 3             \\
\hline\hline		                                                             
SNR$=80$ & $m_1(10^{-3}):0.2\pm 1.4 $ & $5.0\pm 1.2$        & $8.8\pm 0.8$  \\
                     & $c_1(10^{-5}):-7.1\pm 4.7$  & $-6.9\pm 3.8$       & $-4.6\pm 2.6$ \\
\hline
SNR$=60$ & $0.7\pm 1.8$                          & $9.3\pm 1.4$         & $16.6\pm 0.9$ \\
                     & $-8.4\pm 5.9$                         & $-8.2\pm 4.6$        & $-5.1\pm 2.8$ \\
\hline
SNR$=40$ & $1.9\pm 2.6$                          & $21.3\pm 2.0$       & $39.2\pm 1.0$ \\
                     & $-11.1\pm 8.6$                       & $-10.9\pm 6.5$      & $-6.1\pm 3.3$ \\
\hline
SNR$=30$ & $2.9\pm 3.5$                           & $38.0\pm 2.7$       & $72.5\pm 1.2$  \\
                     & $-13.9\pm 11.4$                     & $-13.8\pm 8.6$      & $-7.2\pm 3.9$ \\
\hline
SNR$=20$ & $5.0\pm 5.5$                           & $87.5\pm 4.2$       & $179.5\pm 1.8$ \\
                     & $-19.9\pm 17.8$                      & $-20.5\pm 13.6$   & $-10.0\pm 5.7$ \\
\hline
SNR$=10$ & $10.3\pm 13.3$                       & $451.9\pm 12.3$   & $1529.5\pm 6.8$   \\
                     & $-40.8\pm 43.4$                      & $-56.0\pm 40.1$    & $-34.8\pm 22.2$   \\

\hline\hline		

\end{tabular}
\caption{The multiplicative bias $m_1$ and the additive bias $c_1$ of the first component of the reduced shear measured using the method of this paper. In each data cell, the upper value is $m_1$ in unit of $10^{-3}$, and the lower value is $c_1$ in unit of $10^{-5}$. The results are shown for six choices of SNR. For each SNR, three cases are considered: 1. solely background noise; 2. the background noise and source Poisson noise have equal amplitudes; 3. solely source Poisson noise. }
\label{noise_cases1}
\end{table}

\begin{table}
\centering
\begin{tabular}{cccc}

\hline\hline		
                     &  Case 1                                       & Case 2                 &   Case 3             \\
\hline\hline		                                                             
SNR$=80$ & $m_2(10^{-3}):-0.3\pm 1.4 $  & $4.3\pm 1.2$       & $8.3\pm 0.8$  \\
                     & $c_2(10^{-5}):3.0\pm 4.7$      & $2.2\pm 3.8$       & $1.4\pm 2.6$ \\
\hline
SNR$=60$ & $0.1\pm 1.8$                             & $8.3\pm 1.4$       & $15.9\pm 0.9$  \\
                     & $2.9\pm 5.9$                             & $1.9\pm 4.6$        & $1.1\pm 2.8$ \\
\hline
SNR$=40$ & $0.8\pm 2.6$                             & $19.6\pm 2.0$     & $38.1\pm 1.0$  \\
                     & $2.2\pm 8.6$                             & $1.2\pm 6.5$        & $0.3\pm 3.3$  \\
\hline
SNR$=30$ & $1.4\pm 3.5$                             & $35.7\pm 2.7$      & $70.7\pm 1.2$ \\
                     & $1.1\pm 11.4$                           & $0.2\pm 8.6$        & $-0.5\pm 3.9$ \\
\hline
SNR$=20$ & $2.8\pm 5.5$                             & $84.0\pm 4.2$      & $176.5\pm 1.8$ \\
                     & $-2.7\pm 17.8$                          & $-2.7\pm 13.6$    & $-2.2\pm 5.7$ \\
\hline
SNR$=10$ & $6.7\pm 13.3$                           & $444.5\pm 12.3$  & $1516.7\pm 6.8$  \\
                     & $-26.2\pm 43.4$                        & $-23.7\pm 40.1$   & $-15.3\pm 22.2$ \\

\hline\hline		

\end{tabular}
\caption{Same as table \ref{noise_cases2}, except that it is for multiplicative and additive biases ($m_2$ and $c_2$) of the second component of the reduced shear.}
\label{noise_cases2}
\end{table}

The results in tables \ref{noise_cases1} and \ref{noise_cases2} show that the shear recovery accuracy in our method strongly depends on whether the noise is mostly due to the background or the source. As shown in all the first cases in the tables, the shear measurement errors caused by pure background noise can be cleanly removed with the treatment of \cite{zhang09}. We can almost confirm the sub-percent level accuracy for SNR$\sim 10$ in this case. The statistical errors, though systematically larger than those in other cases of the same SNR due to the nature of the noise treatment of \cite{zhang09}, can be further narrowed down to check the accuracy of the method at even smaller SNR with a larger galaxy ensemble. In the second and third cases of each SNR, the multiplicative biases are clearly all larger than their statistical errors, and reach $\gtrsim 1\%$ level for SNR$\lesssim 60$. To achieve a sub-percent level accuracy, we find that it requires ${\rm SNR}_S\gtrsim 80-100$ in our method. It corresponds to collecting roughly $10^4$ source photons per galaxy. Indeed, the limit on the ${\rm SNR}_S$ is so far the only requirement for achieve very high accuracy in our shear measurement method. 

Physically, the shear recovery errors in case 2 and 3 are caused by the fact that the galaxy shapes are intrinsically modified by their own Poisson noise. It seems likely to smooth out the source Poisson noise if the half-light radius of the source covers enough pixels. However, it is hard to see how useful such a procedure is, because for a given ${\rm SNR}_S$, more pixels within the half-light radius mean a larger Poisson noise in each pixel. In the end, the total source flux should always has a relative uncertainty of order $1/{\rm SNR}_S$. Along this line of thinking, one can even guess a simple relation between ${\rm SNR}_S$ and the amplitude of the multiplicative bias $m$ through the following procedure: 
\begin{eqnarray}
\label{m_SNR}
&&m\sim\frac{\delta g}{g}\sim \frac{\left\langle \delta (f^2)\right\rangle}{\left\langle f^2\right\rangle}\\ \nonumber
&&\left\langle \delta (f^2)\right\rangle\sim\left\langle (f_S+f_N)^2-f_S^2\right\rangle\sim\left\langle 2f_Sf_N+f_N^2\right\rangle\sim\left\langle f_N^2\right\rangle \\ \nonumber
&&\Rightarrow m\sim\frac{\left\langle f_N^2\right\rangle}{\left\langle f_S^2\right\rangle}\sim\frac{1}{{\rm SNR}_S^2}
\end{eqnarray} 
In the above estimation, $g$ stands for the reduced shear, $f$ refers to the total flux within the half-light radius of the source, and the subscripts $S$ and $N$ denote the source and noise. Note that the $\left\langle f_N^2\right\rangle$ due to the background noise is statistically estimated using a neighboring map of pure noise and subtracted in our method. The remaining shear measurement errors therefore solely come from the source Poisson noise. The results in tables \ref{noise_cases1} and \ref{noise_cases2}  indeed indicate that $m\approx 60/{\rm SNR}_S^2$, except when ${\rm SNR}_S\lesssim 10$. Note that this is true in both case 2 and case 3. We caution that the derivation in eq.(\ref{m_SNR}) is not rigorous. It only provides a possible understanding of the results seen in tables \ref{noise_cases1} and \ref{noise_cases2}. We will study source Poisson noise more systematically in a future work.

\section{Summary}
\label{summary}

Based on \cite{zhang08,zhang09,zhang10}, we have established a robust way of measuring the cosmic shear to the second order in accuracy. The method is well defined regardless of the morphologies of the galaxies and the PSF. We have also provided a useful theorem for judging the accuracy of any shear measurement method at the second order based on its properties at the first order.

For our method to achieve the accuracy at sub-percent level, the CCD pixel size is required to be not larger than about $1/3$ of the FWHM of the PSF, regardless of whether the PSF has a power-law or exponential profile at large distances\footnote{For PSFs with strong diffraction spikes, we need to further test the method. This will be done in a future work.}. Using more than $10^7$ mock galaxies of unrestricted morphologies, we have tested the accuracy of this method under different noise conditions.  We find that it is useful to separately discuss the background and source noise for any given SNR. The background noise, which is uncorrelated with the source flux, can be removed in a simple and clean way using the method of \cite{zhang09}. In our simulations with only background noise, the shear measurement errors are found to be less than $1\%$ for SNR as low as $10$, and the conclusion can likely be extended to even smaller SNR with simulations of a larger galaxy ensemble. On the other hand, the source Poisson noise, which strongly couples with the distribution of the source flux, remains to be the main cause of the shear measurement errors in our method. For a sub-percent level accuracy, we require the SNR of the source Poisson noise to be $\gtrsim 80-100$. This corresponds to collecting about $10^4$ source photons per galaxy. The treatment of source Poisson noise is unclear at present, and will hopefully be addressed in a future work.

\acknowledgments{JZ would like to thank the anonymous referees for illuminating comments, Yi Mao for his help on parallel computations, and the Texas Advanced Computing Center for providing High Performance Computing (HPC) resources. JZ is currently supported by the TCC Fellowship of Texas Cosmology Center of the University of Texas at Austin. JZ was previously supported by the TAC Fellowship of the Theoretical Astrophysics Center of UC Berkeley, where part of this work was done. }

\section*{Appendix -- The Relation Between Two Definitions of Cosmic Shears}
\label{appendix}

The convention of defining the cosmic shear/convergence in this paper is different from what is used in \cite{zhang08,zhang09,zhang10}. More specifically, as shown in eq.(\ref{define1}), we define shear/convergence using the following formula:
\begin{eqnarray}
\label{convention2}
\left(\begin{array}{c}
x_1^S \\
x_2^S \end{array}\right)=\left(\begin{array}{cc}
1-\kappa-\gamma_1 & -\gamma_2 \\
-\gamma_2 & 1-\kappa+\gamma_1 \end{array}\right)\left(\begin{array}{c}
x_1^L \\
x_2^L \end{array}\right).
\end{eqnarray}
In \cite{zhang08,zhang09,zhang10}, however, shear/convergence is defined as:
\begin{eqnarray}
\label{convention1}
\left(\begin{array}{c}
x_1^L \\
x_2^L \end{array}\right)=\left(\begin{array}{cc}
1+\kappa'+\gamma_1' & \gamma_2' \\
\gamma_2' & 1+\kappa'-\gamma_1' \end{array}\right)\left(\begin{array}{c}
x_1^S \\
x_2^S \end{array}\right).
\end{eqnarray}

Note that these two conventions are equivalent up to the first order in shear/convergence, but not to the second order. In other words, the values of $\gamma_1$, $\gamma_2$, and $\kappa$ are equal to those of $\gamma_1'$, $\gamma_2'$, and $\kappa'$ respectively only when the second and higher order lensing terms are neglected. It is straightforward to derive the relation between the two conventions based on the following identity:
\begin{equation}
\label{relation_identity}
\left(\begin{array}{cc}
1+\kappa'+\gamma_1' & \gamma_2' \\
\gamma_2' & 1+\kappa'-\gamma_1' \end{array}\right)
=\left(\begin{array}{cc}
1-\kappa-\gamma_1 & -\gamma_2 \\
-\gamma_2 & 1-\kappa+\gamma_1 \end{array}\right)^{-1}. 
\end{equation}
To the second order in shear/convergence, we get:
\begin{eqnarray}
\label{relation}
\kappa'&=&\kappa+\kappa^2+\gamma_1^2+\gamma_2^2,\nonumber \\
\gamma_1'&=&\gamma_1(1+2\kappa),\nonumber \\
\gamma_2'&=&\gamma_2(1+2\kappa).
\end{eqnarray}
As a result, we find:
\begin{equation}
\label{reduced_shear_relation}
\gamma_{1,2}'(1-\kappa')=\gamma_{1,2}(1+\kappa).
\end{equation}
Therefore, for theorists intending to make second order predictions for the cosmic shears, conventions should be explicitly mentioned, as they carry different meanings and consequences at the second order level.

\vskip 1cm

\label{lastpage}

\end{document}